\documentclass[a4paper, 10pt, conference]{ieeeconf}      % Use this line for a4
\IEEEoverridecommandlockouts                              % This command is only
\usepackage{graphicx}      % include this line if your document contains figures
\usepackage{amsfonts}
\usepackage{amsmath}
\usepackage{amssymb}
\usepackage{makeidx}
\usepackage{epsfig}
\usepackage{tikz}
\usepackage{bm}
\usetikzlibrary{arrows}

\def\R{\mathbb {R}}
\def\P{\mathbb {P}}
\def\N{\mathbb {N}}
\def\cN{\mathcal{N}}

\def\tp{^{\rm T}}
\let\oldhat\hat
\renewcommand{\vec}[1]{\bm{#1}}
\renewcommand{\hat}[1]{\oldhat{\mathbf{#1}}}

\newtheorem{theorem}{Theorem}

\newtheorem{definition}{Definition}

\newtheorem{lemma}[theorem]{Lemma}
\newtheorem{proposition}[theorem]{Proposition}
\newtheorem{problem}{Problem}
\newtheorem{remark}{Remark}

\title{\LARGE \bf A sub--optimal solution for optimal control\\of linear systems with unmeasurable switching delays}

\author{A. Cicone, A. D'Innocenzo, N. Guglielmi and L. Laglia
\thanks{The authors are with the Center of Excellence DEWS and with the Department of Information Engineering, Computer Science and Mathematics of the University of L'Aquila, Italy. The research leading to these results has received funding from the Italian Government under Cipe resolution n.135 (Dec. 21, 2012), project \emph{INnovating City Planning through Information and Communication Technologies} (INCIPICT)}
}

\begin{document}

\maketitle

\begin{abstract}
We consider the optimal control design problem for discrete--time LTI systems with state feedback, when the actuation signal is subject to unmeasurable switching propagation delays, due to e.g. the routing in a multi--hop communication network and/or jitter. In particular, we set up a constrained optimization problem where the cost function is the worst--case $\mathcal{L}_2$ norm for all admissible switching delays. We first show how to model these systems as pure switching linear systems, and as main contribution of the paper we provide an algorithm to compute a sub--optimal solution.
\end{abstract}

%%%%%%%%%%%%%%%%%%%%%%%%%%%%%%%%%%%%%%%%%%%%%%%%%%%%%%%%%%%%%%%%%%%%%%%%%%%%%%%%

\section{Introduction} \label{secIntro}

\begin{figure}[ht]
\begin{center}
%\vspace{-1cm}
\includegraphics[width=0.5\textwidth]{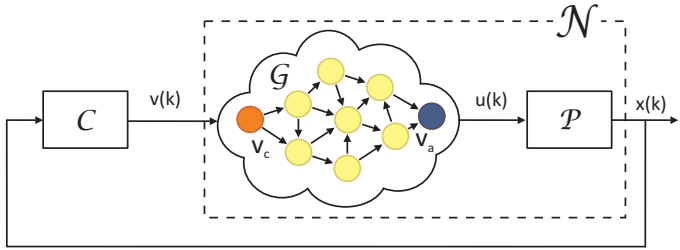}
%\vspace{-0.6cm}
\caption{State feedback control scheme.}\label{controlScheme}
%\vspace{-1.2cm}
\end{center}
\end{figure}

Wireless networked control systems are spatially distributed control systems where the communication between sensors, actuators, and computational units is supported by a wireless multi--hop communication network. The main motivation for studying such systems is the emerging use of wireless technologies for control systems (see e.g. \cite{SongIECON2010} and references therein) and the recent development of wireless industrial control protocols (such as WirelessHART and ISA--100). Although the use of wireless networked control systems offers many advantages with respect to wired architectures (e.g. flexible architectures, reduced installation/debugging/diagnostic/maintenance costs), their use is a challenge when one has to take into account the joint dynamics of the plant and of the communication protocol. Recently, a huge effort has been made in scientific research on Networked Control Systems (NCS) and on the interaction between control systems and communication protocols, see e.g.~\cite{HLA_TAC2009},~\cite{MurrayTAC2009},~\cite{Hespanha2007},~\cite{HeemelsTAC10},~\cite{TabbaraTAC2007}. In general, the literature on NCSs addresses non--idealities (e.g. quantization errors, packets dropouts, variable sampling and delay, communication constraints) as aggregated network performance variables or disturbances, neglecting the dynamics introduced by the communication protocols. To the best of our knowledge, the first integrated framework for analysis and co--design of network topology, scheduling, routing and control in a wireless multi--hop control network has been presented in \cite{AlurTAC11}, where switching systems are used as a unifying formalism for control algorithms and communication protocols and sufficient conditions for stabilizing the plant are provided. {In \cite{DInnocenzoTAC13} the networked system is sampled at the time scale of the period of the communication protocol scheduling, and this makes the system discrete--time linear time--invariant.

In \cite{JungersTAC2016} we refined the above model and considered a networked system sampled at the more accurate time scale of the transmission slots of each communication node: this makes the system discrete--time switching linear, which is a much more difficult mathematical framework. In particular, we assumed in our model that a multi--hop network $\mathcal G$ provides the interconnection between a state--feedback discrete--time controller $\mathcal C$ and a discrete--time LTI plant $\mathcal P$ (see Figure \ref{controlScheme}). The network $\mathcal G$ consists of an acyclic graph where the node $v_c$ is directly connected to the controller and the node $v_a$ is directly connected to the actuator of the plant. As classically done in (wireless) multi--hop networks to improve robustness with respect to node failures we assume that the number of paths that interconnect $\mathcal C$ to $\mathcal P$ is greater than one and that for any actuation data sent from the controller to the plant a unique path is chosen. Each path is characterized by a delay in forwarding the data (see \cite{DInnocenzoTAC13} for details), as a consequence each actuation data will be delayed of a finite number of time steps according to the chosen routing: as a consequence our system is characterized by switching time--varying delays of the input signal, due to routing as well as jitter. Since usually jitter is unpredictable and the routing choice depends on the internal status of the network, i.e. because of node and/or link failures, we considered the choice of the routing path and/or jitter effects as an external disturbance. In practical applications measurability of such disturbance switching signal (that represents the network induced delays at each time step) depends on the protocol used to route data (see \cite{YangWiMesh2011} and references therein for an overview on routing protocols for wireless multi--hop networks). If the controller node $v_c$ of $\mathcal G$ is allowed by the protocol to choose a priori the routing path (e.g. source routing protocols), then we can assume that the controller is aware of the routing path and the associated delay, and therefore can measure the switching signal. If instead the protocol allows each communication node to choose the next destination node according to the local neighboring network status information (e.g. hop--by--hop routing protocols) then we cannot assume that the controller can measure the switching signal. In \cite{JungersTAC2016} we first characterized stabilizability in the case when the switching signal is measurable, then we showed with some examples that when it is unmeasurable it is much harder to decide stabilizability, e.g. it can be necessary to make use of nonlinear controllers for stabilizing a linear plant.

Motivated by these results and examples, we address in this paper the control design problem when the switching signal is unmeasurable. In particular, we set up a constrained optimization problem where the cost function is the worst--case $\mathcal{L}_2$ norm for all admissible switching delays. We first show how to model these systems as pure switching linear systems, and as the main contribution of the paper we provide an algorithm to compute a sub--optimal solution. To the best of our knowledge this problem for linear systems with switching delays has not been addressed from the scientific literature so far. Also, since in our model we take into account the realistic situation where delay variations can be due to both routing in a multi--hop communication network and jitter, the set of admissible switching delays can be very large: indeed jitter can generate small delay variations, but routing via different paths can generate very large delay variations. As a consequence, classical robust control design approaches that are based on the assumption of slight delay variations around a central nominal delay are not applicable to our model. Also, since disturbances due to routing and jitter cannot be generally described by some predictable scheme or mathematical model, methods for estimation and prediction of the delay is generally useless.

Relating our results to the existing literature, we remark that analysis and controller design for systems with time--varying delays have attracted increasing attention in recent years (see e.g. \cite{LiuCTA2006}, \cite{HetelCDC2007}, \cite{HeemelsTAC2010}, \cite{ShaoTAC2011}, \cite{HetelIJC2012}, \cite{KaoAUTOMATICA07} and references therein). In particular, in \cite{HetelSCL2011} it is assumed that the time--varying delay is approximatively known and numerical methods are proposed to exploit this partial information for adapting the control law in real time. Our modeling choice is close to the framework in \cite{HetelCDC2007}. However our setting is more general and realistic in that, differently from \cite{HetelCDC2007}, it allows for several critical phenomena to happen: In our model, control commands generated at different times can reach the actuator simultaneously, their arrival time can be inverted, and it is even possible that at certain times no control commands arrive to the actuator. These phenomena are well known to be responsible for controllability issues in Multi--hop Control Networks. In summary, our contribution is providing a novel sub--optimal approach (which is reasonable, since the problem is still highly non--linear) leveraging the special structure induced by the switching delays and exploiting previous results on the computation of the extremal norm of a set of matrices.

The paper is organized as follows. In Section \ref{secModel} we provide the modeling framework and the problem formulation. In Sections \ref{secOverview} and \ref{secJSR} we first overview and then provide the main results of the paper. In Section \ref{secSimulations} we apply our procedure to an illustrative example. In Section \ref{secConclusions} we provide concluding remarks and open problems for future research.

%%%%%%%%%%%%%%%%%%%%%%%%%%%%%%%%%%%%%%%%%%%%%%%%%%%%%%%%%%%%%%%%%%%%%%%%%%%%%%%%

\section{Model and problem formulation} \label{secModel}

Consider a discrete--time linear time invariant system $\mathcal P$ described by the following difference equation:
$$
\vec{x}(k+1) = A_P \vec{x}(k) + B_P \vec{u}(k), \quad k \geq 0,
$$
where $\vec{x} \in \mathbb{R}^n$, $\vec{u} \in \mathbb{R}^m$, $A_P \in \mathbb R^{n \times n}$, $B_P \in \mathbb R^{n \times m}$. Figure~\ref{controlScheme} illustrates a state--feedback control scheme: the objective is that the state signal $\vec{x}(k)$ is as close as possible to a reference signal $\vec{r}(k) \in \mathbb{R}^n$. In this paper we assume that the reference signal $\vec{r}(k)=\vec{0}$, i.e. we wish to drive the state to its equilibrium point in the origin. The extension to the more general case when the reference signal is a Heaviside step function can be approached using a change of variable and the use of an integrator or a feed--forward term in the control scheme, and is out of the scope of this paper. We assume that the control signal $\vec{v}(k)$ generated by the controller $\mathcal C$ is transmitted to the actuator of the plant $\mathcal P$ via a multi--hop network. The network $\mathcal G$ consists of an acyclic graph $(V,E)$, where the node $v_c \in V$ is directly interconnected to the controller $\mathcal C$, and the node $v_a$ is directly interconnected to the actuator of the plant $\mathcal P$. In order to transmit each actuation data $\vec{v}(k)$ to the plant, at each time step $k$ a unique path of nodes that starts from $v_c$ and terminates in $v_a$ is exploited. As classically done in multi--hop (wireless) networks to improve robustness of the system with respect to node failures, we exploit redundancy of routing paths, therefore the number of paths that can be used to reach $v_a$ from $v_c$ is assumed to be greater than one. To each path a different delay can be associated in transmitting data from $v_c$ to $v_a$, depending on the transmission scheduling and on the number of hops to reach the actuator. Since the choice of the routing path usually depends on the internal status of the network (e.g. because of node and/or link failures, bandwidth constraints, security issues, etc.), we assume that the routing path chosen at each time step $k$ is time--varying. Therefore the actuation command $\vec{v}(k)$ at time $k$ will be delayed, according to the chosen routing path and the jitter, of a finite number of time steps, which we model as a disturbance signal $\sigma(k)$. For the reasons above, we can model the dynamics of a multi--hop control network as follows:
\medskip
\begin{definition}\label{def-general-sdsystem}
The dynamics of the interconnected system $\mathcal N$ can be modeled as
\begin{equation}\label{sdsystem}
\vec{x}_e(k+1) = A \vec{x}_e(k) + B(\sigma(k)) \vec{v}(k), \quad k \geq 0,
\end{equation}
where $\vec{x}_e = [\vec{x}\tp\ \vec{u}_{1}\tp\ \cdots\ \vec{u}_{d_{max}}\tp]\tp \in \mathbb{R}^{n + md_{max}}$, $\vec{u}_{i}, i = 1, \ldots, d_{max}$, $\vec{v} \in \mathbb{R}^m$, $\sigma(k) \in D$, for any $k \geq 0$, is a disturbance switching signal taking value in the set $D \subseteq \{0, 1, 2, \dots, d_{max}\}$, which is the set of possible delays introduced by all routing paths with $d_{max}$ the maximum delay, and
$$
A = \left[
      \begin{array}{lll}
        A_P & B_P & \textbf{0}_{n \times m(d_{max} - 1)} \\
        \textbf{0}_{md_{max} \times n} & \textbf{0}_{md_{max} \times m} & \Gamma_{md_{max} \times m(d_{max} - 1)} \\
      \end{array}
    \right],
$$

$$
\Gamma_{md_{max} \times m(d_{max} - 1)} = \left[
                                          \begin{array}{llll}
                                            \mathbb I_{m} & \textbf{0} & \cdots & \textbf{0} \\
                                            \textbf{0} & \mathbb I_{m} & \cdots & \textbf{0} \\
                                            \vdots & \vdots & \ddots & \vdots \\
                                            \textbf{0} & \textbf{0} & \cdots & \mathbb I_{m} \\
                                            \textbf{0} & \textbf{0} & \textbf{0} & \textbf{0} \\
                                          \end{array}
                                        \right],
$$

%$$
%B(\sigma(k)) = \left(B_P \delta_{0,\sigma(k)} \quad \mathbb I_{m} \delta_{1,\sigma(k)} \ \hdots \ \mathbb I_{m} \delta_{d_{max},\sigma(k)} \right)\tp,
%$$

$$
B(\sigma(k)) = \left(
                 \begin{array}{c}
                   B_P \delta_{0,\sigma(k)} \\
                   \mathbb I_{m} \delta_{1,\sigma(k)} \\
                   \vdots \\
                   \mathbb I_{m} \delta_{d_{max},\sigma(k)}\\
                 \end{array}
               \right),
$$
where $\delta_{i,\sigma(k)}, \forall i \in \{0, \ldots, d_{max}\}, k \geq 0$ is the Kronecker delta and $\mathbb I_{m}$ is the identity matrix of dimension $m$.
\end{definition}

We also assume that the controller is not aware of the switching signal $\sigma(k)$, but it can measure the error signal $\vec{e}(k) = \vec{x}(k) - \vec{r}(k)$, where the reference signal $\vec{r}(k)$ is considered for simplicity constantly zero for every $k\geq0$, as already discussed, hence $\vec{e}(k) = \vec{x}(k)$. As a consequence of these assumptions, we can not assume that the controller can measure the variables $\vec{u}_{1}\ \cdots\ \vec{u}_{d_{max}}$. We can instead assume that the controller keeps memory of the last $d_{max}$ measures of the error signal, namely $\vec{e}(k-1), \ldots, \vec{e}(k-d_{max})$:
\medskip
\begin{definition}\label{def-indep} Assume that the switching signal cannot be measured at any time instant. We define a control law as follows:
\begin{align}\label{eq-vt}
\vec{v}(k)&= K [\vec{e}\tp(k),\ \vec{e}\tp(k-1),\ \dots,\ \vec{e}\tp(k-d_{max})]\tp,\notag \\
K &= [K_0\, K_1 \cdots K_{d_{max}}] \in \mathbb{R}^{m \times n (d_{max}+1)}
\end{align}
\end{definition}
\medskip
Therefore the closed--loop system behavior consists of the extended state space
\begin{align}\label{eq:y_e}
y_e(k) = [&\vec{e}\tp(k),\ \vec{e}\tp(k-1),\ \dots,\ \vec{e}\tp(k-d_{max}),\notag \\
&\ \vec{u}_{1}(k),\ \cdots,\ \vec{u}_{d_{max}}(k)]\tp \in \mathbb{R}^{n (d_{max}+1) + md_{max}}.
\end{align}
\normalsize
If we define $\Sigma$ as the set of all switching signals of infinite length $\sigma=(\sigma(0),\, \sigma(1),\,\ldots\,)$ we can formulate the main goal of our research
\medskip
\begin{problem}\label{pb:MinErr}
Fixed the plant parameters $A_P$ and $B_P$, we want to design the controller $K$ in \eqref{eq-vt} such that the following cost function is minimized:
\begin{equation}\label{eq:MinErr}
\min_{K}\mathcal J = \min_{K}\max\limits_{\sigma \in \Sigma} \|\vec{e}\|^2_{\mathcal L_2} \doteq \min_{K}\max\limits_{\sigma \in \Sigma} \sum\limits_{k=0}^{\infty} \|\vec{e}(k)\|_2^2
\end{equation}
\end{problem}
\medskip

In other words, we want to design a controller that minimizes the worst case $\mathcal L_2$--norm of the error signal over all the possible infinite length switching signals $\sigma\in\Sigma$.

Since the cardinality of $\Sigma$ is infinite, the problem in hand, to the best of the authors' knowledge, cannot be solved exactly. For this reason, in the next Section we develop an approximation method which allows to produce upper bounds for \eqref{eq:MinErr}.

\section{Overview of the proposed method} \label{secOverview}

For a switched control system given by \eqref{sdsystem} and \eqref{eq-vt} we aim to compute a sharp upper bound for the $\ell_2$--norm of the sequence $\{\vec{e}(k)\}_{k \ge 0}$ over all possible switching laws and minimize this quantity with respect to the parameters of the controller $K_0,K_1,\ldots K_{d_{max}}$. A necessary and sufficient condition for the sequence being in $\ell_2$ is that the joint spectral radius of the associated set of matrices is smaller than~$1$: it is well known, in fact, that the maximal rate of growth among all products of matrices from a finite set $\mathcal M$ is given by the Joint Spectral Radius (JSR) of $\mathcal M$, $\rho(\mathcal M)$, which is the generalization of spectral radius to sets of matrices \cite{RotStr,BerWan,LagWan,DauLag1,GugZen08,cicone2015JSRnote,jungers_lncis}.

Motivated by this observation we first minimize the joint spectral radius with respect to $K = [K_0\, K_1 \cdots K_{d_{max}}]$ by a descent algorithm and then estimate the associated $\ell_2$--norm. The minimization of the joint spectral radius can be done directly on the underlying set of matrices to the switched system given by \eqref{sdsystem} and \eqref{eq-vt}. In order to get a sharp estimate of the $\ell_2$--norm we propose an approximation framework based on the identification of a certain polytope norm $\| \cdot \|_P$ (computed as we describe below) and its mathematical relation with the $2$--norm of the variable $\vec{e}(k)$. Since the dynamics of the switched control system given by \eqref{sdsystem} and \eqref{eq-vt} also involves the actuation variable $\vec{u}_i(k), i = 1, \ldots, d_{max}$, such formal relation is hard to state. To overcome this difficulty we derive \eqref{eq:sys2}, which is an equivalent formulation of the switched control system given by \eqref{sdsystem} and \eqref{eq-vt}, that only involves the sequence $\{ \vec{e}(k) \}$, but leads to an augmented formulation that extends the state dimension of the system to $2 d_{\max}+1$. This is done by considering switching sequences of a certain length (for example $2$ in the case $D = \{0,\ 1\}$), which however implies some constraint in the new system (for example the control sequence $00$ cannot be followed by the sequence $11$) and gives a Markovian structure \cite{kozyakin2014berger} to the joint spectral radius problem. This problem can be solved by a suitable lifting of the matrices (recently proposed by Dai and Kozyakyn \cite{dai2014robust,kozyakin2014berger}) of the set which determines an equivalent classical joint spectral radius problem for a new set of matrices. According to a methodology which has been recently developed \cite{protassov1996joint,GWZ05,Guglielmi2013Exact}, we compute the joint spectral radius by determining a polytope extremal norm (which we indicate as $\| \cdot \|_P$) of a set of matrices $\mathcal F$. In order to compute a bound for the $\ell_2$ norm we have to find a constant which relates the computed polytope norm $\| \cdot \|_P$ and the $2$--norm, i.e. $\| \vec{x} \| \le C_P \| \vec{x} \|_P$ and use the fact that for all matrices $F \in \mathcal F$ of the set we have $\| F \|_P = \rho(\mathcal F)$. This can be done efficiently by computing the dual polytope to the one computed by our algorithm. This estimate is useful to bound (by a geometric series) the norm of the sequence $\{ e(k) \}_{k \ge \eta}$ for $\eta$ an arbitrarily chosen natural number. For the first part of the sequence, which in many cases gives the most important contribute to the estimate, we compute directly all matrix products of degree smaller than $\tau \le \eta$ in the product semigroup and then use another estimate (still based on the construction of certain polytopes) to estimate the products of length $\tau+1 \le k \le \eta$ (an intermediate regime), which usually provides better bounds with respect to the geometric estimate, and which we shall describe in the following sections.

\section{Main results} \label{secJSR}

From now on we focus our attention on the case of $m=1$, thus $\vec{e}(k),\ \vec{e}(k-1),\ \dots,\ e(k-d_{max}) \in \mathbb R^n$ and $\vec{u}_{1},\ \cdots\,\ \vec{u}_{d_{max}} \in \mathbb R$. We assume that the initial conditions are given by an arbitrary $\vec{e}(0)$ while $\forall k < 0, e(k) = \mathbf{0} \in \mathbb R^n$.

Before stating the main results we present a technical Lemma.

\begin{lemma}\label{lemma_1} Given the system \eqref{sdsystem} and the control law \eqref{eq-vt},

Then, for every $k\geq 0$
\small
\begin{equation}\label{eq:sys2}
\vec\varepsilon_e(k+1) = N_{\sigma(k), \ldots, \sigma(k-d_{max})} \vec\varepsilon_e(k),
\end{equation}
\normalsize
with $\vec\varepsilon_e(k) \doteq [\vec{e}\tp(k),\ \vec{e}\tp(k-1),\ldots, \vec{e}\tp(k-2d_{max})]\tp$, $N_{\sigma(k), \ldots, \sigma(k-d_{max})} \in \mathbb{R}^{\left(n (2d_{max}+1)\right)^2}$, and
\begin{equation}\label{eq:N}
\cN=\left\{ N_{i_0 \ldots i_{d_{max}}}  \right\}_{i_0, \ldots, i_{d_{max}} \in D}, \ |\cN| = |D|^{d_{max}+1}
\end{equation}
\end{lemma}
\vskip 2mm
\begin{proof}
We start rewriting the system (\ref{sdsystem}), together with the control law \eqref{eq-vt}, as follows:
\begin{equation}\label{eq:sys_n}
y_e(k+1)=
%&=\left[
%  \begin{array}{l}
%    \vec{e}(k+1) \\
%    \vec{e}(k) \\
%    \vdots \\
%    \vec{e}(k-d_{max}+1) \\
%    \vec{u}_1(k+1) \\
%    \vdots \\
%    \vec{u}_{d_{max}}(k+1) \\
%  \end{array}
%\right] \notag\\
%&= M(\sigma(k)) \left[
%  \begin{array}{l}
%    \vec{e}(k) \\
%    \vec{e}(k-1) \\
%    \vdots \\
%    \vec{e}(k-d_{max}) \\
%    \vec{u}_1(k) \\
%    \vdots \\
%    \vec{u}_{d_{max}}(k) \\
%  \end{array}
%\right]=
M(\sigma(k)) y_e(k),
\end{equation}
for every $k\geq 0$, where
\begin{eqnarray*}
  \vec{e}(k+1) &=& A_P \vec{e}(k) + B_P \vec{u}_1(k) + B_P \delta_{0, \sigma(k)} \vec{v}(k)\\
\vec{e}(k) &=& \vec{e}(k)\\
&\vdots & \\
\vec{e}(k+1-d_{max}) &=& \vec{e}(k-(d_{max}-1))\\
\vec{u}_1(k+1) &=& \vec{u}_2(k) + \delta_{1, \sigma(k)} \vec{v}(k)\\
&\vdots & \\
\vec{u}_{d_{max}-1}(k+1) &=& \vec{u}_{d_{max}}(k) + \delta_{d_{max}-1, \sigma(k)} \vec{v}(k)\\
\vec{u}_{d_{max}}(k+1) &=& \delta_{d_{max}, \sigma(k)} \vec{v}(k).
\end{eqnarray*}
%\begin{align*}
%
%\end{align*}
Then, we note that the dynamics of $\vec{u}_1(k)$ are given by
\vskip -2mm
\begin{align*}
\vec{u}_1(k) =& \delta_{1, \sigma(k-1)} \vec{v}(k-1) + \delta_{2, \sigma(k-2)} \vec{v}(k-2) \\
&+ \ldots + \delta_{d_{max}, \sigma(k-d_{max})} \vec{v}(k-d_{max}),
\end{align*}
\normalsize
hence, by \eqref{eq-vt},
\small
\begin{align*}
\vec{e}(k+1) &= A_P \vec{e}(k) + B_P \sum\limits_{i=0}^{d_{max}} \delta_{i, \sigma(k-i)} \vec{v}(k-i) \notag\\
&=  A_P \vec{e}(k) + B_P \sum\limits_{i=0}^{d_{max}} \delta_{i, \sigma(k-i)} \sum\limits_{j=0}^{d_{max}} K_{j} \vec{e}(k-i-j)
\end{align*}
\normalsize
Then the conclusion follows.
%where
%\begin{equation*}
%N_{0 0} = \left[
%            \begin{array}{lll}
%              A_P + B_P K_1 & B_P K_2 & 0 \\
%              \mathbb I_{n} & 0 & 0 \\
%              0 & \mathbb I_{n} & 0 \\
%            \end{array}
%          \right],
%N_{0 1} = \left[
%            \begin{array}{lll}
%              A_P + B_P K_1 & B_P (K_1+K_2) & B_P K_2 \\
%              \mathbb I_{n} & 0 & 0 \\
%              0 & \mathbb I_{n} & 0 \\
%            \end{array}
%          \right],
%\end{equation*}
%\begin{equation*}
%N_{1 0} = \left[
%            \begin{array}{lll}
%              a & 0 & 0 \\
%              \mathbb I_{n} & 0 & 0 \\
%              0 & \mathbb I_{n} & 0 \\
%            \end{array}
%          \right],
%N_{1 1} = \left[
%            \begin{array}{lll}
%              A_P & B_P K_1 & B_P K_2 \\
%              \mathbb I_{n} & 0 & 0 \\
%              0 & \mathbb I_{n} & 0 \\
%            \end{array}
%          \right]
%\end{equation*}
\end{proof}

\begin{remark} It is clear that not all products of matrices in $\cN$ are allowed. More precisely, if at time $k$ the switching matrix defining the system dynamics is given by $N_{ \bar\sigma_k, \ldots, \bar\sigma_{k-d_{max}}}$, then at time $k+1$ the set of allowed matrices is given by the set $\{ N_{ \sigma(k+1), \bar\sigma_k, \ldots, \bar\sigma_{k-d_{max+1}}} : \bar\sigma(k+1) \in D \}$ with cardinality $|D|$. In Figure~\ref{fig:Adj} we represent, for instance, the admissible left products  for $D = \{0,\ 1\}$ by means of edges of a graph whose nodes correspond to matrices in $\cN$. For example the edge from node $N_{10}$ to node $N_{11}$ represents the product $N_{11}N_{10}$. Note that the outdegree (i.e. the number of outgoing edges) of each vertex is equal to $2 = |\{0,1\}|$: it is easy to see that, in the general case, the outdegree of each vertex is equal to $|D|$.

We observe that in the continuous time case the idea of admissible transitions between subsystems goes under the name of \emph{constrained switching} and has been studied, for example, in \cite{Souza2014}.

\begin{figure}[ht]
\begin{center}
\begin{tikzpicture}[->,>=stealth',shorten >=1pt,auto,node distance=2cm,
                    thick,main node/.style={circle,draw,font=\sffamily\Large\bfseries}]

  \node[main node] (1) {$N_{00}$};
  \node[main node] (2) [right of=1] {$N_{01}$};
  \node[main node] (3) [below of=1] {$N_{10}$};
  \node[main node] (4) [below of=2] {$N_{11}$};

  \path[every node/.style={font=\sffamily\small}]
    (1) edge [bend right] node[right] {} (3)
        edge [loop left] node {} (1)
    (2) edge [bend right] node [right] {} (1)
        edge node [left] {} (3)
        %edge [loop left] node {0.4} (2)
        %edge [bend right] node[left] {0.1} (3)
    (3) edge node [right] {} (2)
        edge [bend right] node[right] {} (4)
    (4) edge [bend right] node [left] {} (2)
        edge [loop right] node {} (4);
        %edge [bend right] node[right] {0.2} (1);
\end{tikzpicture}
\caption{Admissible products graph.}\label{fig:Adj}
\end{center}
\end{figure}
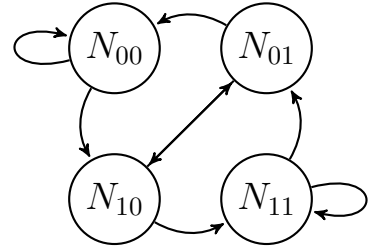
\end{remark}

Given $\vec\varepsilon_e(k+1)= N_{\sigma(k), \ldots, \sigma(k-d_{max})}\vec\varepsilon_e(k)$, $\forall k\geq 0$, then we see immediately that $\sum\limits_{k=0}^{\infty} \|\vec{e}(k)\|_2^2 = \sum\limits_{k=0}^{\infty} \frac{\|\vec\varepsilon_e(k)\|_2^2}{2d_{max}+1}$.

The main result of this paper can be summarized as follows

\begin{proposition}\label{MainProp}
Given \eqref{sdsystem} and the control law \eqref{eq-vt},

There $\exists\ C > 0$ such that

\scriptsize
\begin{align}\label{eq:MainProp}
  &\min_{K}\mathcal J \leq  \frac{1}{2d_{max}+1} \min_{K} \left(\max_{\sigma_{[0,\tau]}\in\Sigma}\sum\limits_{k=0}^{\tau} \|\vec\varepsilon_e(k)\|^2_2+ \right.\notag\\
  &\left.\sum_{k=\tau+1}^{\eta} \alpha^{k} \max_{{\overline{\vec{v}}}\in\overline V_k} \parallel {\overline{\vec{v}}} \parallel_{2}^{2}+C^2 \rho(\overline{\cN})^{2\eta+2} \frac{1}{1-\rho(\overline{\cN})^{2}} \left\|\vec\varepsilon_e(1)\right\|_{2}^{2}\right)
\end{align}
\normalsize
where $\sigma_{[a,b]}$ is the string of switching signals contained in $\sigma$ from position $a \in \mathbb{N}_0$ to $b \in \mathbb{N}_0$, $\overline V_k$, $k\geq 0$ are properly constructed polytopes, and $\rho(\overline{\cN})$ is the Joint Spectral Radius of the set of matrices obtained as a special lifting of the set $\cN$.
\end{proposition}

To prove this Proposition we need a few intermediate results.

\begin{lemma}\label{lemma_2}
Given the function $\mathcal J$, defined in Problem \ref{pb:MinErr},

Then
\begin{equation}\label{eq:MinNorm}
  \min_{K}\mathcal J \equiv \frac{1}{2d_{max}+1}\min_{K} \max_{\sigma \in \Sigma} \sum_{k=0}^{\infty}\| P_{k}(\sigma) \vec{\varepsilon}_e(0) \|_{2}^{2}
\end{equation}
where $P_{k}(\sigma)$ is the product of matrices in $\cN$ associated with the first $k$ values in $\sigma$ for $k\geq 0$,  and $P_{0}(\sigma)=\mathbb I_n$ is the identity matrix.
\end{lemma}
\begin{proof}
From the definition of $\mathcal J$ and the extended error vectors $\vec{\varepsilon}_e(k)$ it follows
\small
\begin{equation}\label{eq:MinErr2}
\min_{K}\mathcal J =\frac{1}{2d_{max}+1}\min_{K}\max\limits_{\sigma \in \Sigma} \sum\limits_{k=0}^{\infty} \|\vec\varepsilon_e(k)\|_2^2.
\end{equation}
\normalsize
Hence, from Lemma \ref{lemma_1}, \eqref{eq:MinNorm} follows directly.
\end{proof}

\begin{remark} Considering that the cardinality of $\Sigma$ is infinite, we can make the computational complexity feasible by computing upper bounds for \eqref{eq:MinNorm} using the Joint Spectral Radius theory. In particular it holds true that
\begin{equation}\label{eq:JSRlimNorm}
\rho(\mathcal M)=\lim_{k\rightarrow \infty} \sup_{P\in \P_k(\mathcal M)} \|P\|^{1/k}
\end{equation}
for any matrix norm $\|\cdot\|$, where $\P_k(\mathcal M)$ is the set of all products of length $k$ of matrices in $\mathcal M$.
Furthermore
\begin{equation}\label{eq:JSRnorm}
    \rho(\mathcal M)=\inf_{\|\cdot\|\in {\mathcal O}}\sup_{A\in {\mathcal M}}\|A\|
\end{equation}
where ${\mathcal O}$ denotes the set of all possible operator norms. We observe that, if the set $\mathcal M\in\R^{d \times d}$ is irreducible, which means that only the trivial subspaces $\{0\}$ and $\R^{d \times d}$ are invariant under the action all the matrices in $\mathcal M$, then the inf in \eqref{eq:JSRnorm} is actually achieved for some induced norm $\|\cdot\|_*$ which is said to be extremal for the set $\mathcal M$ \cite{Bara88}.

In the problem under study, however, not all the products of matrices in $\cN$ are allowed, as explained above. Admissible ones induce a Markovian structure \cite{kozyakin2014berger} of the switching sequences. For instance, in the case of $D = \{0,\ 1\}$ the admissible left products can be represented as the paths of the graph plotted in Figure \ref{fig:Adj}. Recently it has been shown that the maximal rate of growth among all and only the admissible products of matrices from a finite set $\mathcal M$, which is called the Constrained Joint Spectral Radius (CJSR), can be computed through the evaluation of the classical JSR of a proper lifting of the set $\mathcal M$ \cite{dai2014robust,kozyakin2014berger}.
\end{remark}

To make use of the JSR theory in order to produce upper bounds for \eqref{eq:MinNorm} we need first the following result:

\begin{lemma}\label{lemma_3} Given an induced polytope norm $\|\bullet \|_P$ for a finite dimensional space $S$, whose unit ball is the polytope $P$, then $\|A\|_2\leq C \|A \|_P$ for any $A\in S$, where $C\leq R(P)R(P^*)$ and $R(P)$, $R(P^*)$ are the circumradii of $P$ and its dual, $P^*$, respectively.
\end{lemma}
\begin{proof} Given the finite dimensional space $S$, the induced polytope norm $\|\bullet \|_P$ and the polytopic unit ball $P$, by the equivalence of norms in finite dimensional spaces, there exists a constant of equivalence $C$ such that $\|A\|_2\leq C \|A \|_P$ for any $A\in S$. Furthermore it is easy to check that $C\leq \frac{R(P)}{r(P)}$, where $R(P)$ and $r(P)$ are the circumradius and inradius of $P$, respectively. By \cite[Theorem 1.2]{gritzmann1992inner}, it follows that $C\leq R(P)R(P^*)$.
\end{proof}
\begin{remark}
 There are many well studied algorithms that allow to build the dual $P^*$ of a polytope $P$, like for instance the one described in \cite{avis2000revised}. Furthermore we point out that the circumradius of a polytope $P$ and its dual, $P^*$, can be evaluated easily as the maximum length of vertices on the boundary of $P$ and $P^*$.
\end{remark}

\begin{lemma}\label{lemma_4} Consider \eqref{eq:MinNorm} and a set of matrices $\overline{\cN}$ produced as a lifting of the matrices in $\cN$, as showed by means of an example in the next Section. Assuming that the set $\overline{\cN}$ is irreducible and it does admit an extremal polytope norm $\|\bullet\|_*$ such that $\|\overline N\|_*\leq \rho(\overline{\cN})$ for any $\overline N\in\overline{\cN}$,

Then
\scriptsize
\begin{equation}\label{eq:MinNormUB}
\min_{K}\mathcal J\leq \frac{1}{2d_{max}+1}\min_{K} \left(C^2\sum_{k=1}^{\infty} \rho(\overline{\cN})^{2k} \left\|\vec\varepsilon_e(0)\right\|_{2}^{2}+\left\|\vec{\varepsilon}_e(0)\right\|_{2}^{2}\right)
\end{equation}
\normalsize
where $C$ is the constant of equivalence such that $\|A\|_2\leq C \|A\|_*$, for any $A\in\R^{\left(n (2d_{max}+1)|D|^{d_{max}}\right)^2}$, whose existence is guaranteed by the equivalence of norms in finite dimensional spaces.
\end{lemma}
\begin{proof}
Following the ideas presented in \cite{kozyakin2014berger}, we consider a new set $\overline{\cN}$ containing $|D|^{d_{max}+1}$ matrices of size $\left(n (2d_{max}+1)|D|^{d_{max}} \right)^2$ which are given by a special lifting of the matrices in $\cN$. The two sets are such that the JSR of $\overline{\cN}$ equals the CJSR of the set $\cN$. In the next Section we give an example of such a lifting. Here we observe that, due to the structure of the matrices in $\overline{\cN}$, each product $\overline P_{k}(\sigma)$, associated with the first $k\geq 0$ values in $\sigma\in\Sigma$, does have only one non--zero block of size $\left(n (2d_{max}+1)\right)^2$. Therefore, given any switching sequence $\sigma\in\Sigma$, it follows that $\left\| \overline{P}_{k}(\sigma)\right\|_{2}^{2} = \left\| {P}_{k}(\sigma) \right\|_{2}^{2}$. Furthermore, always due to the structure of matrices in $\overline{\cN}$, we can consider $|D|^{d_{max}}$ lifted vectors $\overline{\vec\varepsilon}_e(0)^{(h)}$, $h\in\{1,\,\ldots,\ |D|^{d_{max}}\}$, such that each of them contains $|D|^{d_{max}}$ blocks of size $n (2d_{max}+1) \times 1$ which are all zero except one that equals $\vec\varepsilon_e(0)$. As a consequence $\left\|\overline{\vec\varepsilon}_e(0)^{(h)}\right\|_{2}^{2} =\left\|\vec{\varepsilon}_e(0) \right\|_{2}^{2}$ for every  $h\in\{1,\,\ldots,\ |D|^{d_{max}}\}$, and, for each product $\overline P_{k}(\sigma)$, there is a $h\in\{1,\,\ldots,\ |D|^{d_{max}}\}$ such that $\left\| \overline{P}_{k}(\sigma) \overline{\vec\varepsilon}_e(0)^{(h)}\right\|_{2}^{2} = \left\| {P}_{k}(\sigma)\vec\varepsilon_e(0)\right\|_{2}^{2}$.
Whence \eqref{eq:MinNorm} can be rewritten as
\scriptsize
\begin{align}\label{eq:MinNormBar}
&(2d_{max}+1)\min_{K}\mathcal J =\min_{K} \max_{\sigma \in \Sigma} \left(\sum_{k=1}^{\infty}\left\| \overline{P}_{k}(\sigma) \overline{\vec\varepsilon}_e(0) \right\|_{2}^{2}+\left\|\vec{\varepsilon}_e(0)\right\|_{2}^{2}\right)\notag\\
&\leq \min_{K} \max_{\sigma \in \Sigma} \left(\sum_{k=1}^{\infty}\left\| \overline{P}_{k}(\sigma)\right\|_{2}^{2} \left\|\vec\varepsilon_e(0) \right\|_{2}^{2} + \left\| \vec{\varepsilon}_e(0) \right\|_{2}^{2}\right)
\end{align}
\normalsize
%where  $\overline P_{k}(\sigma)$ is the product of matrices in $\overline\cN$ associated with the first $k$ values in $\sigma$, for $k\geq 0$.

By assumption we have an extremal polytope norm, so by Lemma \ref{lemma_3} and the submultiplicativity of the 2--norm we get
\eqref{eq:MinNormUB} where $C$ can be derived as explained in Lemma \ref{lemma_3}.
\end{proof}

\begin{remark}
In general the set $\overline{\cN}$ is irreducible, so we can look for a polytopic extremal norm $\|\cdot\|_*$ such that, for any $\overline N\in\overline{\cN}$, $\|\overline N\|_*\leq \rho(\overline{\cN})$. This can be, in theory, an extremely hard problem. However, in general, we can build an extremal polytope norm following the algorithms described in \cite{cicone2010finiteness} and \cite{Guglielmi2013Exact}.
\end{remark}

\begin{remark}
A necessary condition to have a finite $\|\vec\varepsilon_e\|^2_{\mathcal L_2}$ is that the steady state of the extended error signal $\vec\varepsilon_e(k)$ is zero, which is equivalent to require that the joint spectral radius $\rho(\overline{\cN})$ is strictly smaller than 1.
\end{remark}

Based on this observation we can formulate the last Lemma

\begin{lemma}\label{lemma_5} Given \eqref{eq:MinNorm}, the matrices $A_P,\ B_P$ and assuming that there exists $K$ such that the corresponding set $\overline{\cN}$ has $\rho(\overline{\cN})<1$. Then
\scriptsize
\begin{align}\label{eq:J_UB3}
&\min_{K}\mathcal J\leq  \frac{1}{2d_{max}+1}\min_{K}\left(\max_{\sigma_{[0,\tau]}\in\Sigma}\sum\limits_{k=0}^{\tau} \|\vec\varepsilon_e(k)\|^2_2+\right.\notag\\
&\left.\sum_{k=\tau+1}^{\eta} \alpha^{k} \max_{\overline{\vec{v}}\in\overline V_k} \parallel \overline{\vec{v}} \parallel_{2}^{2}+C^2 \rho(\overline{\cN})^{2\eta+2} \frac{1}{1-\rho(\overline{\cN})^{2}} \left\|\vec\varepsilon_e(0)\right\|_{2}^{2}\right)
\end{align}
\normalsize
where $\overline V_k$ are polytopes computed applying products of matrices in $\overline \cN$ to the initial vector $\vec{\varepsilon}_e(0)$, and $\alpha>\rho(\overline{\cN})$.
\end{lemma}
\begin{proof} From \eqref{eq:MinNorm} and Lemma \ref{lemma_4} it follows that Problem~\ref{pb:MinErr} %, which is an equivalent formulation of Problem \ref{pb:MinErr},
admits the following upper bound
\scriptsize
\begin{align}\label{eq:J_UB2}
\max\limits_{\sigma \in \Sigma} \|\vec\varepsilon_e\|^2_{\mathcal L_2} =& \max_{\sigma \in \Sigma} \Big(\sum_{k=0}^{\tau-1}\| {P}_{k}(\sigma) \vec{\varepsilon}_e(0) \|_{2}^{2}+
\sum_{k=\tau}^{\eta-1}\| P_{k}(\sigma) \vec{\varepsilon}_e(0)\|_{2}^{2} \notag\\
&+\sum_{k=\eta}^{\infty}\| P_{k}(\sigma) \vec{\varepsilon}_e(0) \|_{2}^{2}\Big)\notag\\
\leq&  \Big(\max_{\sigma_{[0,\tau]}\in\Sigma}\sum\limits_{k=0}^{\tau} \|\vec\varepsilon_e(k)\|^2_2+ \max_{\sigma_{[\tau+1,\eta]}\in\Sigma}\sum\limits_{k=\tau+1}^{\eta} \|\vec\varepsilon_e(k)\|^2_2\notag\\
&+ C^2 \rho(\overline{\cN})^{2\eta+2} \frac{1}{1-\rho(\overline{\cN})^{2}} \left\|\vec\varepsilon_e(0)\right\|_{2}^{2}\Big).
\end{align}
\normalsize

For the terms from $\tau+1$ to $\eta$ we can construct the following approximation.

Given the set $\overline\cN$, which by assumption has JSR strictly less then 1, we consider the scaling $\widetilde{\cN}=\overline\cN/ \alpha$ such that $\rho(\widetilde{\cN})>1$. We can use, for instance, $\alpha= {\max_{\overline A\in\overline{\cN}}(\overline A)}/{\beta}$ for any $\beta > 1$.
Hence
\scriptsize
\begin{align}\label{eq:Poly}
&\max_{\sigma \in \Sigma} \sum_{k=0}^{\infty}\parallel \vec{\varepsilon}_e(k) \parallel_{2}^{2}= \max_{\sigma \in \Sigma} \sum_{k=0}^{\infty}\parallel \alpha^{k} \widetilde{ P}_{k}(\sigma)\overline{\vec\varepsilon}_e(0)^{(h)} \parallel_{2}^{2}\notag\\
&\leq \sum_{k=0}^{\infty}\max_{\widetilde{ P}_{k} \in \P_k(\widetilde{\cN})}\parallel \alpha^{k} \widetilde{ P}_{k}\overline{\vec\varepsilon}_e(0)^{(h)} \parallel_{2}^{2}\leq \sum_{k=0}^{\infty} \alpha^{k} \max_{\overline{\vec{v}}\in\overline V_k} \parallel \overline{\vec{v}} \parallel_{2}^{2}
\end{align}
\normalsize
where $\widetilde{P}_{0}$ is the identity matrix, $\overline V_k$ is the set of vertices of the polytope produced at step $k$ applying products of length $k$ of matrices in $\widetilde{\cN}$ to the vector $\overline{{\vec\varepsilon}}_e(0)^{(h)}$, and $h\in\{1,\,\ldots,\ |D|^{d_{max}}\}$ is such that $\left\| \overline{P}_{k}(\sigma) \overline{{\vec\varepsilon}}_e(0)^{(h)}\right\|_{2}^{2} = \left\| {P}_{k}(\sigma) {\vec\varepsilon}_e(0)\right\|_{2}^{2}$.

Therefore
\begin{align}
\max_{\sigma_{[0,\eta]}\in\Sigma}\sum\limits_{k=0}^{\eta} \|{\vec\varepsilon}_e(k)\|^2_2\leq &\max_{\sigma_{[0,\tau]}\in\Sigma}\sum\limits_{k=0}^{\tau} \|{\vec\varepsilon}_e(k)\|^2_2+\notag\\
&\sum_{k=\tau+1}^{\eta} \alpha^{k} \max_{\overline{\vec{v}}\in\overline V_k} \parallel \overline{\vec{v}} \parallel_{2}^{2}
\end{align}
and the conclusion follows.
\end{proof}

\begin{remark}
We observe that the penultimate inequality in \eqref{eq:Poly} is due to the selection, for each fixed $k$, of the maximum error norm instead of considering the maximum of the sum of all error norms over any possible sequence of switching signals $\sigma \in \Sigma$. We are overestimating the summation, but we do not need to find anymore the infinite sequence $\sigma$ of switching signals that maximize the cost function.
\end{remark}
\begin{remark}
The last inequality in \eqref{eq:Poly} follows from the properties of the polytope which is built applying to an initial vector $\overline{{\vec\varepsilon}}_e(0)^{(h)}$, $h\in\{1,\,\ldots,\ |D|^{d_{max}}\}$, all the products of matrices in $\widetilde{\cN}$, where $\rho\left(\widetilde{\cN}\right)>1$ due to the choice of the scaling value $\alpha$.
In particular, since $\rho\left(\widetilde{\cN}\right)>1$, whenever the produced vectors span the entire space, which is true in general, then some of them will tend to grow, in the limit, as fast as $\rho\left(\widetilde{\cN}\right)$.
Furthermore, whenever a vector $\overline{\vec{v}}$ ends up being inside the polytope $\overline V_k$ produced at step $k$, then for any step $n>k$ all the vectors generated from $\overline{\vec{v}}$ they will be for sure inside or at most on the boundary of the polytopes $\overline V_n$. Hence if, for any $k\geq 1$, we get rid of all the vectors that are not on the boundary, then $\max_{\overline{\vec{v}}\in\overline V_k} \parallel \overline{\vec{v}} \parallel_{2}^{2}$ ends up being an overestimation of the quantity $\max_{\widetilde{P}_{k} \in \P_k(\widetilde{\cN})}\parallel \widetilde{P}_{k}\overline{\vec\varepsilon}_e(0)^{(h)} \parallel_{2}^{2}$ since $\overline V_k$ can include vertices computed at previous steps instead of all and only the possible vertices associated with products of length $k$. On the other hand we have the advantage of drastically reducing at each step $k$ the number of vertices that we need to study, as we show in the next Section.
\end{remark}

We are now ready to prove Proposition \ref{MainProp}

\begin{proof}
Using Lemmas from \ref{lemma_1} to \ref{lemma_5}  the conclusions of Proposition  \ref{MainProp} follow naturally.
\end{proof}

\begin{remark} Regarding the choice of the values $\tau$ and $\eta$. The first $\tau$ terms in \eqref{eq:J_UB3} have to be computed using a exhaustive analysis of all the switching signals with length smaller or equal to $\tau$, whose number increases exponentially with respect to $\tau$. Clearly the higher the $\tau$ the more accurate is this approximation of the first $\tau$ terms, but the higher is also the computational time.

Regarding $\eta$ value, the higher $\eta$ is the more polytopes $\overline V_k$ have to be constructed, but also the higher is the accuracy of this bound with respect to the bound obtained exploiting \eqref{eq:MinNormUB}.

In conclusion, one can choose $\eta$ such that the third term of \eqref{eq:MainProp} is negligible, and then choose $\tau$ as large as possible (to increase the overall accuracy) until the computational complexity (which increases exponentially with $\tau$) remains feasible.
\end{remark}
\section{Numerical Examples} \label{secSimulations}

In this Section we present some numerical examples of the application of the above technique.
We concentrate our analysis to cases with $D = \{0,\ 1\}$, but we point out that the proposed method allows to deal with any possible set of delays $D$.

As we mentioned in the previous Section, the maximal rate of growth of admissible products of matrices from the finite set $\cN$ \eqref{eq:N}, the Constrained Joint Spectral Radius of $\cN$, can be computed through the evaluation of the classical JSR of a proper lifting of $\cN$ \cite{kozyakin2014berger}.

In the case of $D = \{0,\ 1\}$ the lifted set is
\begin{equation}\label{eq:Nhat}
    \widehat{\cN}=\left\{\widehat{N}_{0 0},\, \widehat{N}_{1 0},\,\widehat{N}_{0 1},\,\widehat{N}_{1 1} \right\}\subset\mathbb{R}^{12n \times 12n}
\end{equation}
where
\small
\begin{equation*}
\widehat N_{0 0} = \left[
            \begin{array}{llll}
              N_{0 0} & 0 & 0 & 0\\
              0 & 0 & 0 & 0\\
              N_{0 0} & 0 & 0 & 0\\
             0 & 0 & 0 & 0\\
            \end{array}
          \right],
\widehat N_{0 1} = \left[
            \begin{array}{llll}
              0 & N_{0 1} & 0 & 0\\
              0 & 0 & 0 & 0\\
              0 & N_{0 1} & 0 & 0\\
             0 & 0 & 0 & 0\\
            \end{array}
          \right]
\end{equation*}
\begin{equation*}
\widehat N_{1 0} = \left[
            \begin{array}{llll}
            0 & 0 & 0 & 0\\
             0 & N_{1 0} & 0 & 0\\
              0 & 0 & 0 & 0\\
              0 & N_{1 0} & 0 & 0\\
            \end{array}
          \right],
\widehat N_{1 1} = \left[
            \begin{array}{llll}
              0 & 0 & 0 & 0\\
              0 & 0 & 0 & N_{1 1}\\
              0 & 0 & 0 & 0\\
             0 & 0 & 0 & N_{1 1}\\
            \end{array}
          \right]
\end{equation*}
\normalsize
whose JSR $\rho(\widehat{\cN})$ equals the CJSR of set $\cN$.
Furthermore the unitary matrix $U=\frac{1}{\sqrt{2}}\left[\left(
                  \begin{array}{cc}
                    \mathbb I_{6n} & \mathbb I_{6n} \\
                    \mathbb I_{6n} & -\mathbb I_{6n} \\
                  \end{array}
                \right)
\right]$ block triangularizes the set $\widehat{\cN}$ which becomes $\widehat{\cN}'$, such that $\rho(\widehat{\cN})$ equals
$\rho(\widehat{\cN}')$ where
\scriptsize
\begin{equation*}\label{eq:Ntilde}
    \widehat{\cN}'=\left\{\left[
            \begin{array}{ll}
              \overline{N}_{h k} & (-1)^{h}\overline{N}_{h\ k} \\
             0 & 0 \\
            \end{array}
          \right]\right\}\subset\mathbb{R}^{12n \times 12n}\; \textrm{ with } h,k\in\{0,\ 1\}.
\end{equation*}
\normalsize
Given the block triangular structure of this set it follows that  $\rho(\widehat{\cN}')$ equals
$\rho(\overline{\cN})$ where
\scriptsize
\begin{equation*}\label{eq:Nbar_D2}
    \overline{\cN}=\left\{\left[
            \begin{array}{ll}
              N_{0 0} & 0 \\
             0 & 0 \\
            \end{array}
          \right]\!, \left[
            \begin{array}{ll}
              0 & N_{0 1} \\
             0 & 0 \\
            \end{array}
          \right]\!,  \left[
            \begin{array}{ll}
              0 & 0 \\
             N_{1 0} & 0 \\
            \end{array}
          \right]\!, \left[
            \begin{array}{ll}
              0 & 0 \\
             0 & N_{1 1} \\
            \end{array}
          \right] \right\}.
\end{equation*}
\normalsize
So we work directly with the lifted set $\overline{\cN}$ such that $\rho(\overline{\cN})$ equals the CJSR of the set $\cN$, defined in \eqref{eq:N}.

We point out that the previous procedure can be directly extended to the case of delays $D \subseteq \{0, 1, 2, \dots, d_{max}\}$.

%Using a block representation with blocks of size $n (2d_{max}+1) \times n (2d_{max}+1)$, each matrix $\overline N_{\sigma(k), \ldots, \sigma(k-d_{max})}\in\overline{\cN}$ contains $|D|^{d_{max}}$ times $|D|^{d_{max}}$ blocks. All of them equal zero except the block at row $\sum_{j=2}^{d_{max}}\sigma(k-d_{max}+j)|D|^{j-1}+\sigma(k-d_{max}+1)+1$ and column $\sum_{j=1}^{d_{max}-1}\sigma(k-d_{max}+j)|D|^j+\sigma(k-d_{max})+1$ which equals $N_{\sigma(k), \ldots, \sigma(k-d_{max})}$.

In this paper we focus on the case $n=1$ which already presents the main features of the general case $n\in\N$, and without loss of generality we assume that $b=1$. Furthermore, to avoid trivial cases, we consider values of $a$ such that $|a|\geq 1$.

We observe that, from the numerical results we produced, for any $a\leq -1$  and for any $K$ the system appears to be non controllable. Hence we focus our attention on cases where $a\geq 1$. We consider, in particular, $a = 1.1$ and $1.5$.

\begin{remark}
The equations \eqref{eq:MainProp} have to be minimized over all possible values $K$. From \eqref{eq:MinNormUB} it is clear that there are only two quantities that can influence the right end side of \eqref{eq:MainProp} which are $\rho(\overline{\cN})$ and the equivalence constant $C$. We observe that, for $n=1$, the value of $C$ does not depend much on the quantity $K$ and that the JSR of a set of matrices is locally Lipschitz continuous with respect to any change in the entries of the matrices in the set \cite{wirth2005generalized}. Therefore, in order to obtain tighter upper bounds from  \eqref{eq:MainProp}, we can apply any standard minimization technique valid for continuous functions and numerically minimize the quantity $\rho(\overline{\cN})$ over $K$'s.
\end{remark}

\subsubsection{Case $a = 1.1$}

As explained in the previous sections we compute an upper bound for Problem \ref{pb:MinErr}. In particular we first compute, using a greedy method, $\max_{\sigma_{[0,\tau]}\in\Sigma}\sum\limits_{k=0}^{\tau} \|{\vec\varepsilon}_e(k)\|^2_2$ for a feasible value $\tau$ which, in this example, is set to 9. In Figure \ref{fig:n1_a1.1} (left) we plot in solid blue the values $\|{\vec\varepsilon}_e(k)\|^2_2$ for $0\leq k\leq 9$. Then, after obtaining the values $K_0\approx -0.6085$ and $K_1\approx 0.0941$ from the minimization of the joint spectral radius of the set $\cN$ by means of the default Matlab code \verb"fminsearch", we build the corresponding polytopes $\overline V_k$ and evaluate the quantity $\sum_{k=\tau+1}^{\eta} \alpha^{k} \max_{\overline{\vec{v}}\in\overline V_k} \parallel \overline{\vec{v}} \parallel_{2}^{2}$, ref. Figure \ref{fig:n1_a1.1} (left) black dotted curve. For $a=1.1$ the number of vertices in the polytope $\overline V_k$ after 50 steps is simply 166 versus a theoretical number of $4^{50}$ vertices. We observe that the number of vertices of $\overline V_k$, for all the values of $a\geq 1$ we tested, is always almost constant, around 140, and independent on the step number $k$. Finally, following \cite{GWZ05}, we construct a polytopic extremal norm for the set $\overline{\cN}$. We compute the constant of equivalence of the extremal polytope norm for $\overline{\cN}$ with the 2--norm, as explained in the previous section, and produce the values plotted using magenta crosses in Figure \ref{fig:n1_a1.1} (left). We point out that for $0\leq k \leq 9$ the total greedy method             estimate is approximately  $33.1$, while the one by polytope method is around $52.2$. Finally we observe that in this case the tail of the upper bound estimate, which is given by the third term in right hand side of \eqref{eq:MainProp}, is smaller than $10^{-17}$.

\begin{figure}[ht]
\centering
\begin{minipage}[b]{.48\linewidth}
  \centering
  \centerline{\includegraphics[width=\linewidth]{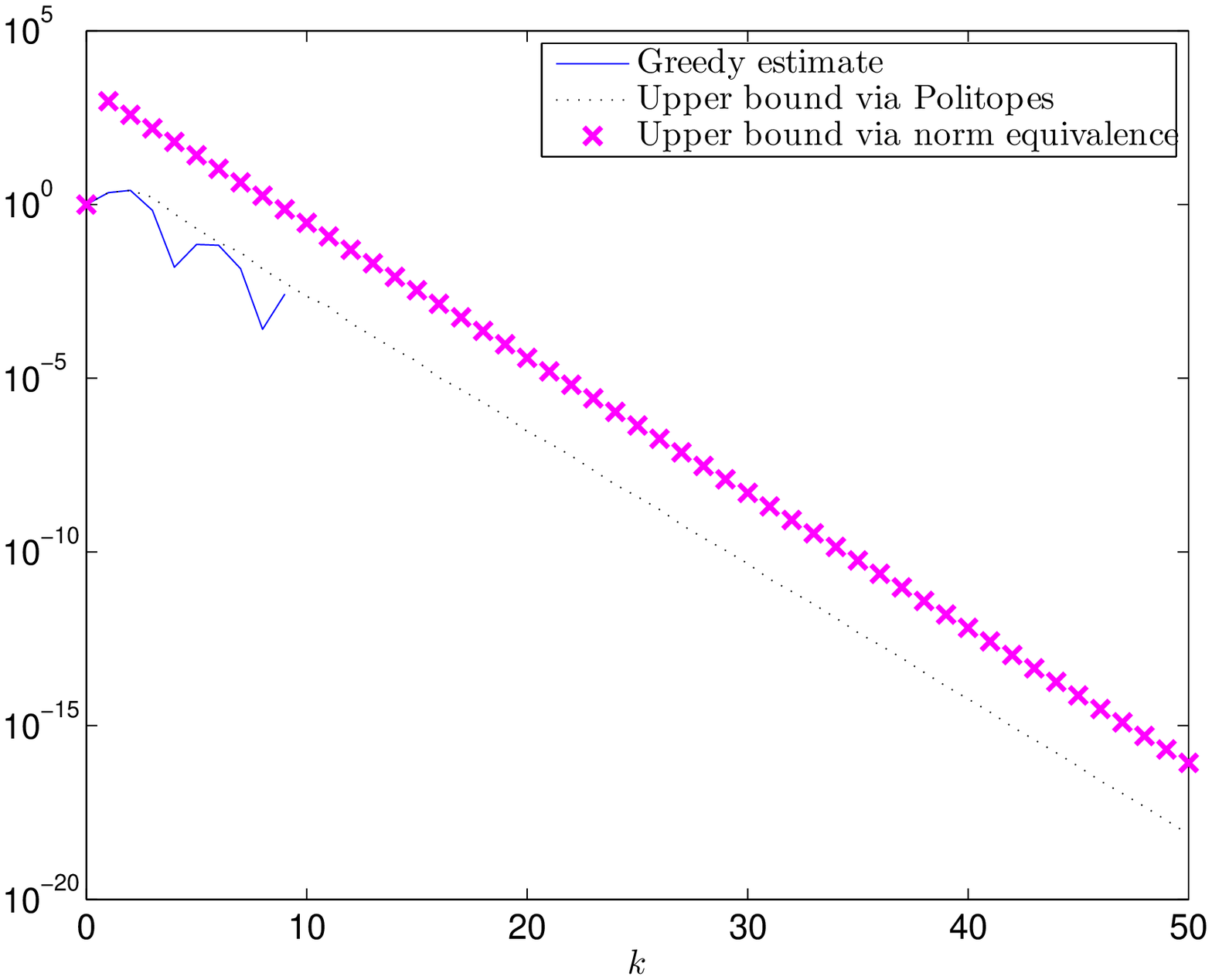}}
%  \vspace{1.5cm}
%\small\centerline{-----}\medskip
\end{minipage}
\hfill
\begin{minipage}[b]{0.48\linewidth}
  \centering
  \centerline{\includegraphics[width=\linewidth]{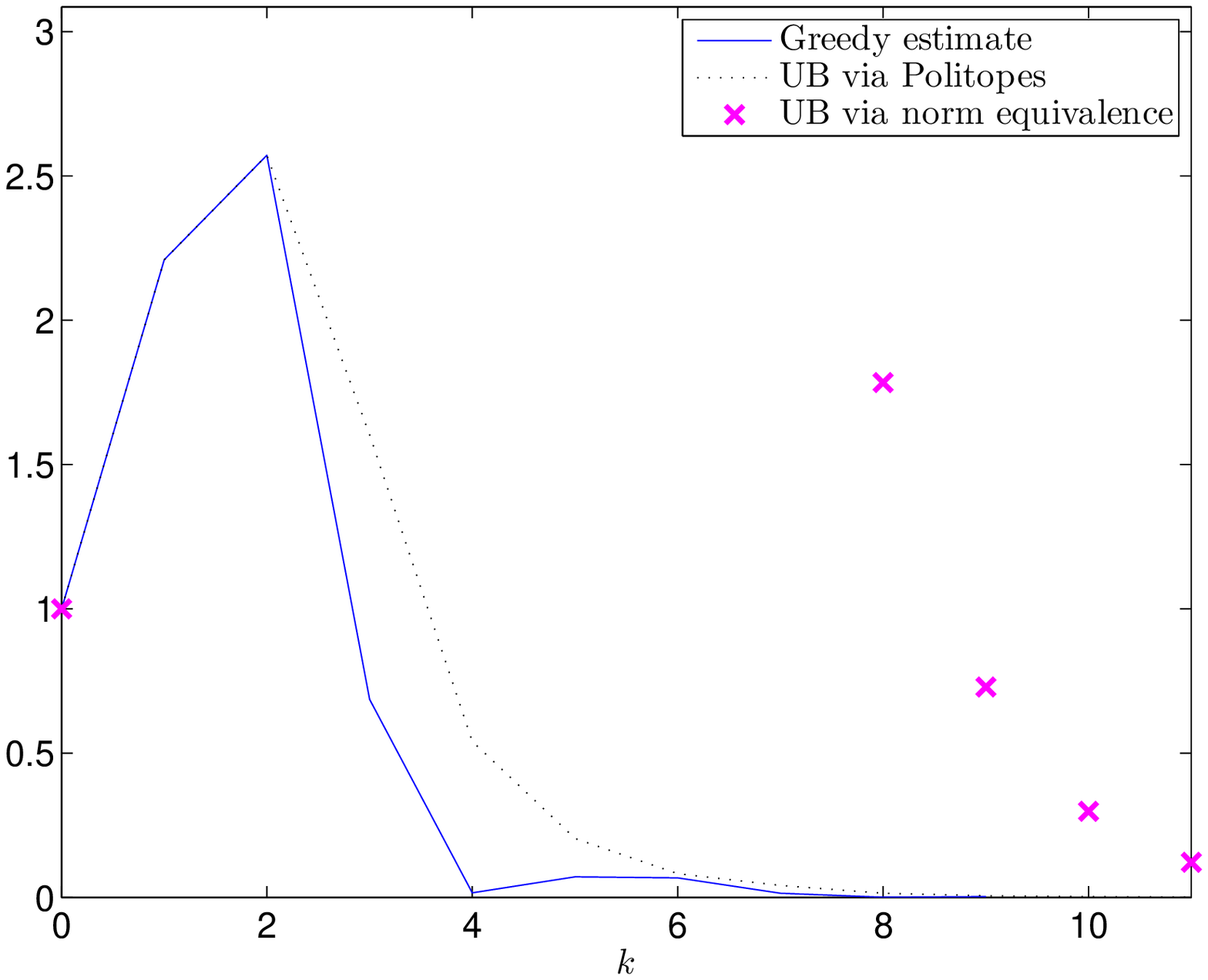}}
%  \vspace{1.5cm}
%\small\centerline{-----}\medskip
\end{minipage}
%
%\vspace*{-0.3cm}
\caption{Upper bounds in \eqref{eq:MainProp} computed for $a=1.1$.}\label{fig:n1_a1.1}
\end{figure}

%
%Case a = 1.10
%
% Estimates of the energy
%
% For k from 0 to 9
%
% Greedy method             33.13420018709749
% Polytope method           52.22467542026522
% Norms equivalence method  7118.36620929970560
%
%
% For k from 10 to 50
%
% Polytope method           0.01879812986981
% Norms equivalence method  2.28944342808189
%
%
% For k after 51
%
% Norms equivalence method  0.00000000000000

\subsubsection{Case $a = 1.5$}

As before, we run a greedy method up to $\tau=9$, Figure \ref{fig:n1_a1.1} (left) solid blue curve. From the minimization of the joint spectral radius of $\cN$ we get the values $K_0\approx  -0.9047$ and $K_1\approx 0.1430$ , for $10\leq k \leq \tau$ we build the polytopes $\overline V_k$, as showed in Figure \ref{fig:n1_a1.1} (left) black dotted curve. This time the number of vertices in $V_{50}$ steps is 156 versus, again, a theoretical number of $4^{50}$ vertices. Finally we compute the constant of equivalence of the polytopic extremal norm with the 2--norm which is around $90.7$,  magenta crosses in Figure \ref{fig:n1_a1.1} (left).  In this example for $0\leq k \leq 9$ the total greedy method estimate is approximately  $106.9$, while the polytopes one is around $163.6$. Finally the third term in right hand side of \eqref{eq:MainProp} in this case is smaller than $10^{-7}$.

\begin{figure}[ht]
\centering
\begin{minipage}[b]{.48\linewidth}
  \centering
  \centerline{\includegraphics[width=\linewidth]{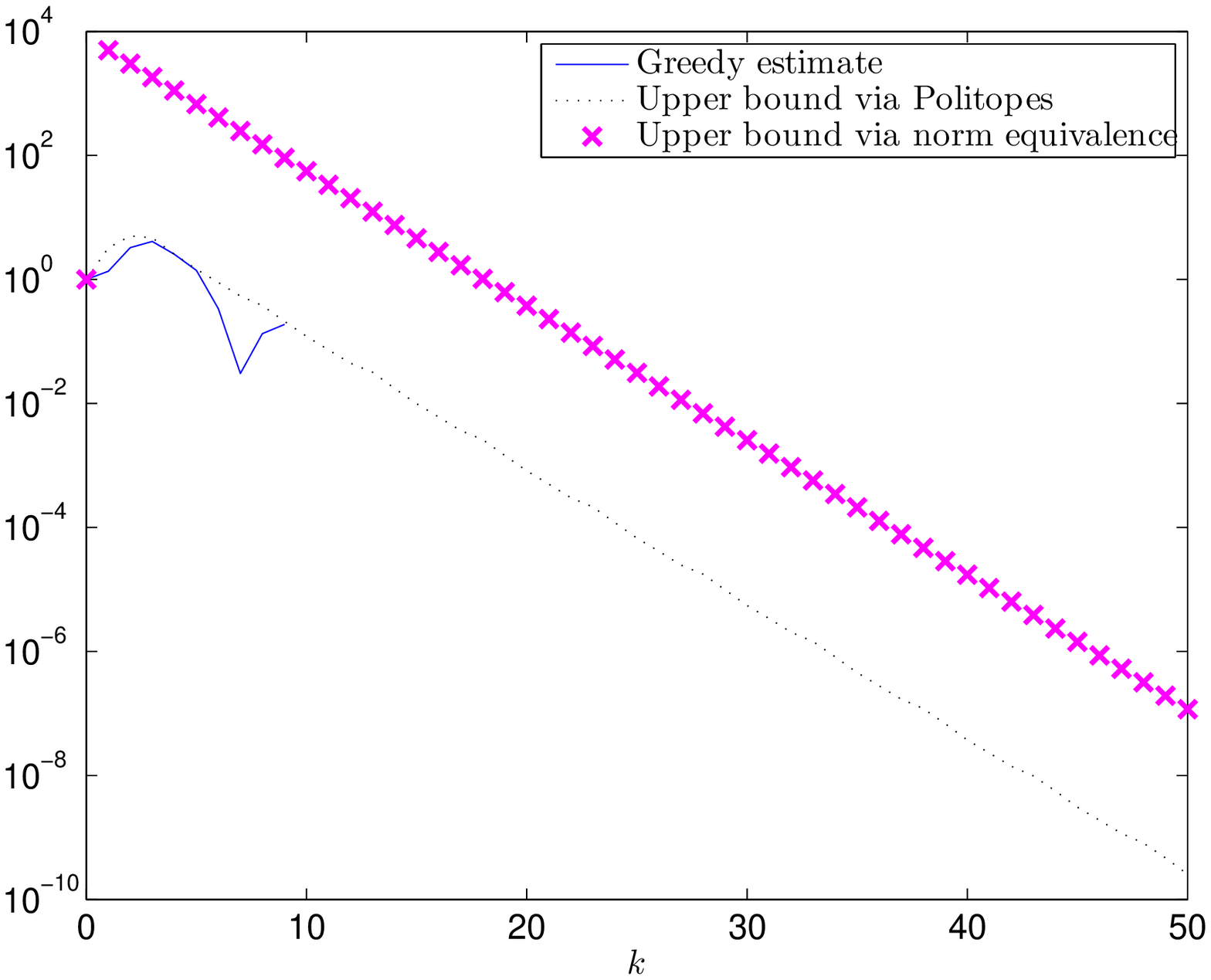}}
%  \vspace{1.5cm}
%\small\centerline{-----}\medskip
\end{minipage}
\hfill
\begin{minipage}[b]{0.48\linewidth}
  \centering
  \centerline{\includegraphics[width=\linewidth]{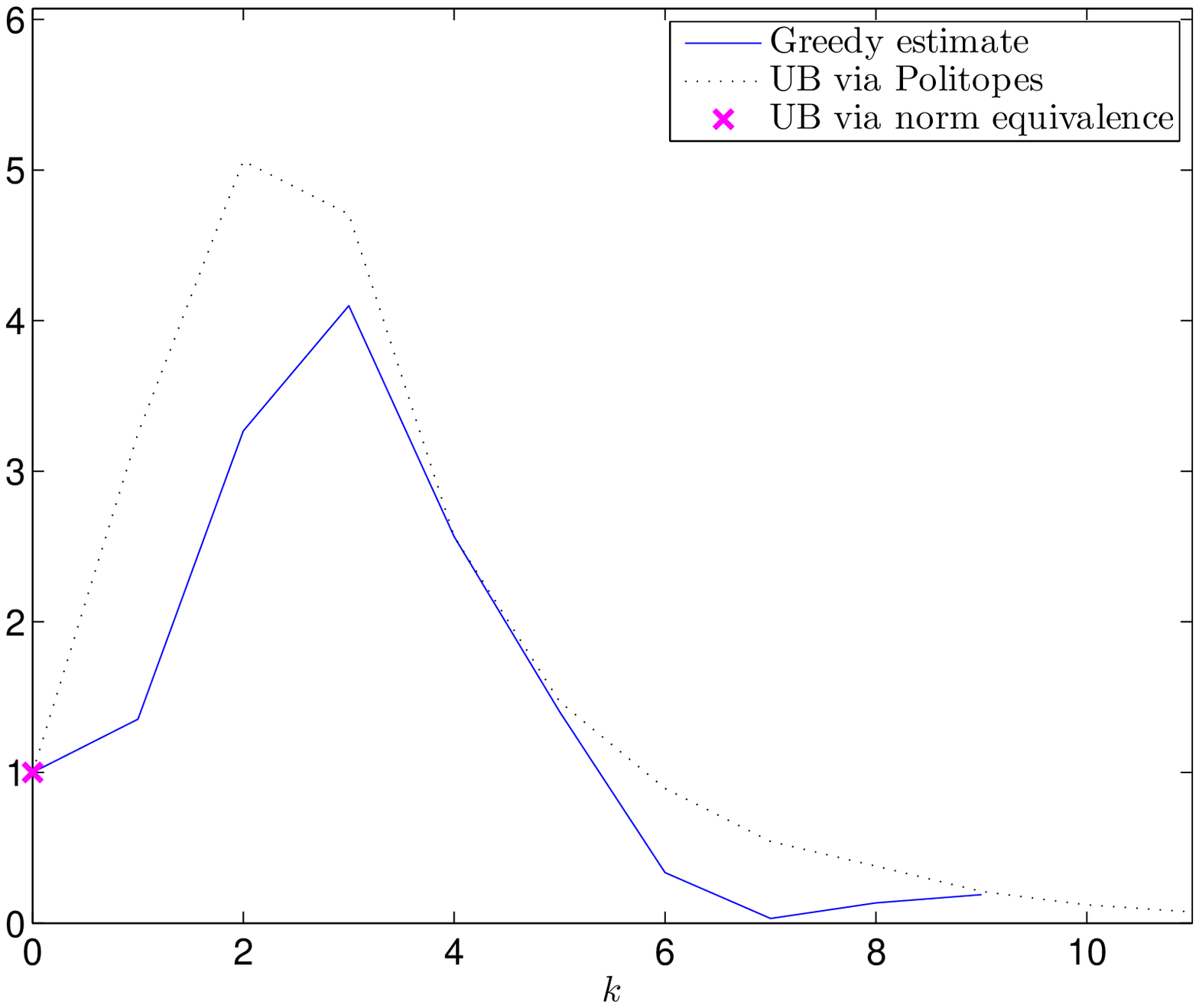}}
%  \vspace{1.5cm}
%\small\centerline{-----}\medskip
\end{minipage}
%
%\vspace*{-0.3cm}
\caption{Upper bounds in \eqref{eq:MainProp} computed for $a=1.5$.}\label{fig:n1_a1.5}
\end{figure}

% Case a = 1.50
%
% Estimates of the energy
%
% For k from 0 to 9
%
% Greedy method             106.85733243294260
% Polytope method           163.63785799172831
% Norms equivalence method  82339.70343136368300
%
%
% For k from 10 to 50
%
% Polytope method           2.55775337550966
% Norms equivalence method  1141.04492141429360
%
%
% For k after 51
%
% Norms equivalence method  0.00000008153023

These two examples show that the construction of the polytopes $\overline V_k$, which allows to speed up the calculations by means of a drastic reduction of the number of vertices at every $k$, allows to compute an upper bound for the error signal $\min_K\mathcal J$ which proves to be close to the greedy estimates in all the tests we ran. Furthermore the calculation of $C$, the constant of equivalence between norms, allows to complete the estimate of this upper bound in finite time.

\section{Conclusions and future work} \label{secConclusions}

We addressed the optimal control design problem for discrete--time LTI systems with state feedback, when the actuation signal is subject to unmeasurable switching propagation delays, due to e.g. the routing in a multi--hop communication network and/or jitter. Solving this problem for general switching systems is a challenging and open problem. Our contribution is providing a sub--optimal approach (which is reasonable, since the problem is still highly non--linear) leveraging the special structure induced by the switching delays and exploiting previous results on the computation of the extremal norm of a set of matrices.

In future work we will more deeply investigate the accuracy of our approach with respect to its computational complexity (related to the variables $\tau$, $\eta$ and $C$), extend our methods for solving the LQR problem for our mathematical model, and apply our results to realistic case studies. Also, we will extend our results to the output-feedback setting: in this case the technical challenges strongly depend on the fact that the communication network allows timestamping or not.

%%%%%%%%%%%%%%%%%%%%%%%%%%%%%%%%%%%%%%%%%%%%%%%%%%%%%%%%%%%%%%%%%%%%%%%%%%%%%%%%

%\nocite{*}

\def\cprime{$'$} \newcommand{\noopsort}[1]{} \newcommand{\singleletter}[1]{#1}

%==============================================================================================================
%==============================================================================================================

\end{document}